\documentclass{amsart}
\usepackage{amsfonts}
\usepackage{amssymb}
\usepackage{times}

\newtheorem{lemma}{Lemma}

\newtheorem{remark}{Remark}
\newtheorem{theorem}{Theorem}
  \newtheorem{corollary}{Corollary}
  \newtheorem{proposition}{Proposition}

  \newtheorem{definition}{Definition}

\newfont{\hueca}{msbm10}
\def\hu #1{\hbox{\hueca #1}}\def\hu #1{\hbox{\hueca #1}}

\begin{document}
\date{}
\author{A.J. Calder\'{o}n, D.M. Cheikh}

\title{Poisson color algebras of arbitrary degree}

\maketitle
\begin{abstract}
A  Poisson algebra is a Lie algebra endowed with a  commutative
associative product in such a way that the Lie and associative
products are compatible via a Leibniz rule. If we part from a Lie
color
 algebra, instead of a Lie algebra,   a graded-commutative  associative product and a graded-version   Leibniz rule
 we get  a so-called  Poisson color algebra (of degree zero). This concept can be extended to any degree  so as to
 obtain the class of Poisson color algebras of arbitrary degree. This class turns out  to be a wide class of
algebras containing the ones of Lie color algebras (and so Lie
superalgebras and Lie algebras),  Poisson algebras, graded Poisson
algebras, $z$-Poisson algebras, Gerstenhaber algebras and Schouten
algebras  among others classes of algebras. The
  present paper is devoted to the study of
 the structure of Poisson color algebras of arbitrary degree,  with restrictions neither
  on the dimension nor the base field.

\end{abstract}

\medskip



\section{Introduction}

On
 the one hand, we recall that Batalin-Vilkovisky (BV) formalism
was introduced in physics as a way of dealing with gauge theories,
being of special interest in the study of  path integrals in
quantum field theory. It can also be seen as a procedure for the
quantization of physical systems with symmetries in the Lagrangian
formalism (see \cite{Ca4, Ca13, Geo5}). BV formalism is just an
example of application of graded Poisson algebras of integer
degree. As another example, we note that it is possible to recover
Hamiltonian mechanics from the coordinate space of the theory by
making use of graded Poisson algebras (\cite{Cor}). We can
enumerate many more applications (see \cite{2,4,Geo2,Geo5,1}), but
we refer to \cite{Catta} to a good review on this matter.

\begin{definition}\rm\label{graded}
Let   ${\mathcal P} = \bigoplus\limits_{z \in {\mathbb Z}}
{\mathcal P}_z$ be a ${\mathbb Z}$-graded vector space   endowed
with a bilinear product $\{\cdot, \cdot\}$ such that
 $$\{{\mathcal P}_{z},{\mathcal P}_{z^{\prime}}\} \subset
 {\mathcal P}_{z+z^{\prime}+ z_0}$$ for any $z, z^{\prime} \in {\mathbb
 Z}$ and a fixed $z_0 \in {\mathbb Z}$, and satisfying  the  identities
  $$\hbox{$\{x,y\}=-(-1)^{(| x|+z_0)( | y|+z_0)}\{y,x\}$,}$$
and
  $$\hbox{$\{x,\{y,z\}\}=\{\{x,y\},z\}+(-1)^{(| x|+z_0)( | y|+z_0)}\{y,\{x,z\}\}$}$$
   for any homogeneous elements $x \in
 {\mathcal P}_{| x| }$, $y \in
 {\mathcal P}_{| y|}$ and $z \in
 {\mathcal P}_{| z|}$.
   ${\mathcal P}$ is called a {\it graded Poisson algebra of
 degree $z_0$}  if it is also  endowed with an associative product, denoted  by
 juxtaposition, such that
  $${\mathcal P}_{z_1}{\mathcal P}_{z_2} \subset
 {\mathcal P}_{z_1+z_2} $$ for any $z_1,z_2 \in {\mathbb Z}$,  and
 satisfies
$$ \hbox{$xy=(-1)^{| x| | y|}yx$}$$  and  $$\hbox{$\{x,yz\}=\{x,y\}z+(-1)^{(| x|+z_0)| y|}y\{x,z\}$ }$$
 for any $x \in
 {\mathcal P}_{| x|}$,  $y \in
 {\mathcal P}_{| y|}$ and $z \in
 {\mathcal P}_{| z|}$.
\end{definition}

\medskip

In the case $z_0=0$ we deal with  {\it even Poisson algebras}
while in the case $z_0=1$ we are dealing with {\it Gerstenhaber
algebras}.

\bigskip

 On the other hand, we also recall that Lie color algebras were introduced in \cite{1} as  a
generalization of Lie superalgebras and hence of Lie algebras.
Since then, this kind of algebras has been an object of constant
interest in
 mathematics,  (see  \cite{Dmitri,  Price, Zhang, Xueme, Kaiming} for recent references),
 being also valuable  the important role they
play in theoretical physics, especially in conformal field theory
and supersymmetries (\cite{C1,C2,JMP2,JMP1}).

\begin{definition}\rm
Let ${\hu K}$ be  an arbitrary  field    and fix an
 abelian  group $(G,+)$. A {\it
skew-symmetric bicharacter}  of   $G$ is a map $$\epsilon : G
\times G \longrightarrow {\hu K}\setminus \{0\}$$ satisfying
$$\epsilon(g_1,g_2)= \epsilon(g_2,g_1)^{-1},$$
$$\epsilon(g_1,g_2+g_3)=
\epsilon(g_1,g_2)\epsilon(g_1,g_3),$$ for any $g_1,g_2,g_3 \in G$.
\end{definition}


\begin{definition}\rm
Let $(G,+)$ be an abelian group,  $\epsilon$ a skew-symmetric
bicharacter  of   $G$ and $${\mathcal P} = \bigoplus\limits_{g \in
G} {\mathcal P}_g$$  a $G$-graded ${\hu K}$-vector space.  We
shall say that ${\mathcal P}$ is a {\it Lie color algebra} if it
is endowed with a bilinear product $\{\cdot, \cdot\}$ satisfying
$$\{{\mathcal P}_{g},{\mathcal P}_{h}\} \subset
 {\mathcal P}_{g+h}$$ for any $g,h \in G,$
 and such that
$$\{x,y\} = -\epsilon({|x| },|y|)\{y,x\}$$
 and
  $$\{x,\{y,z\}\} = \{\{x,y\},z\} + \epsilon(|x|,
  |y|)\{y,\{x,z\}\}$$
 for any homogeneous elements $x \in
 {\mathcal P}_{| x| }$, $y \in
 {\mathcal P}_{| y|}$ and $z \in
 {\mathcal P}_{| z|}$.
\end{definition}
Lie superalgebras (and so Lie algebras)  are examples of Lie color
algebras by considering  $G = {\hu Z}_2$ and $\epsilon(i,j) =
(-1)^{ij},$ for any $i,j \in {\hu Z}_2.$

\medskip

Now we have to note that
 another  class  of  Poisson-type  algebras similar to the one of graded Poisson algebras of degree $z_0$ in  Definition
\ref{graded} but replacing  the group ${\mathbb Z}$  by ${\mathbb
Z}_2$ has been considered in the literature. This kind of algebras
are known as {\it even and odd Poisson superalgebras}, depending
on taking degree $\bar{0}$ or degree $\bar{1}$, being  of interest
in studying, for instance,
  two-dimensional supergravity and three-dimensional systems
  (\cite{2,super,4,1}). However, as we know,  there is not a category  in the literature which allows
  us to combine a graded bracket of degree $g_0\in G$ and a graded commutative
  associative product via a graded Leibniz identity when the
  group $G$ is an arbitrary abelian group. In the present
  paper we will introduce such a  notion by starting from a degree $g_0$ generalization of a Lie color
  algebra.
\begin{definition}\rm\label{colorpo}
Let $(G,+)$ be an abelian group,  $\epsilon$ a skew-symmetric
bicharacter  of   $G$ and $${\mathcal P} = \bigoplus\limits_{g \in
G} {\mathcal P}_g$$  a $G$-graded ${\hu K}$-vector space   endowed
with a  bilinear product $\{\cdot, \cdot\}$ satisfying
$$\{{\mathcal P}_{g},{\mathcal P}_{h}\} \subset
 {\mathcal P}_{g+h+g_0}$$ for any $g,h \in G$ and a fixed $g_0 \in
 G$,
 and such that
 $$\{x,y\} = -\epsilon({|x| +g_0},|y|+g_0)\{y,x\}\hspace{4.4cm} ({\rm Anticonmutativity})$$
 and
  $$
\{x,\{y,z\}\} = \{\{x,y\},z\} + \epsilon(|x|+g_0,
|y|+g_0)\{y,\{x,z\}\}\hspace{0.85cm} ({\rm Jacobi \hspace{0.2cm}
Identity})$$
 for any homogeneous elements $x \in
 {\mathcal P}_{| x| }$, $y \in
 {\mathcal P}_{| y|}$ and $z \in
 {\mathcal P}_{| z|}$.
  It is said that ${\mathcal P}$ is a {\it  Poisson color algebra of
 degree $g_0$}, if it is also  endowed with an associative product, denoted  by
 yuxtaposition, such that
  $${\mathcal P}_{g}{\mathcal P}_{h} \subset
 {\mathcal P}_{g+h} $$ for any $g,h \in G$,  and
 satisfies
$$xy=\epsilon({|x| },|y|)yx \hspace{6.75cm} ({\rm Conmutativity})$$  and
 $$\{x,yz\}=\{x,y\}z+\epsilon({|x|+g_0 },|y|)y\{x,z\}\hspace{3.4cm}({\rm Leibniz \hspace{0.2cm}
Identity})$$
 for any $x \in
 {\mathcal P}_{| x|}$,  $y \in
 {\mathcal P}_{| y|}$ and $z \in
 {\mathcal P}_{| z|}$.
\end{definition}

\medskip

  This class of algebras turns out  to be a wide one  containing those  of Lie color algebras (and so Lie
superalgebras and Lie algebras),  Poisson algebras, graded Poisson
algebras, $z$-Poisson algebras, Gerstenhaber algebras
(\cite{G1,G2,G3}),  and Schouten algebras (\cite{S1,S2,F2}),
 among
other classes of algebras, being these classes of algebras of
increasing interest in mathematical physics, especially in Hamiltonian
and Lagrangian dynamics and mechanics.  Hence Poisson color
algebras of
 degree $g_0$ allow us to  treat all
of these classes of algebras from a common view point and extend
their formalisms to non-necessarily ${\mathbb Z}$-graded or
${\mathbb Z}_2$-graded contexts.  We also note that the case of
degree 0 has been previously considered in \cite{F1} for the
 case of Banach algebras, in the study of a color extension of
 Hamiltonian formalism. Also a geometric approach to the ideas of
 \cite{F1} can be found in \cite{F2}, where it is presented
 a Poisson geometry in this context.

 \medskip

 The usual regularity concepts will be understood in the graded
  sense. That is, a  {\it subalgebra}  of a Poisson color algebra ${\mathcal P}$
  of arbitrary
 degree
 is a graded linear subspace ${\mathcal Q}$
satisfying $\{{\mathcal Q},{\mathcal Q} \} + {\mathcal Q}{\mathcal
Q} \subset {\mathcal Q}$.   An {\it ideal } ${\mathcal I}$ of
${\mathcal P}$ is a subalgebra satisfying $\{{\mathcal
I},{\mathcal P} \} +\{{\mathcal P},{\mathcal I} \}+ {\mathcal
I}{\mathcal P}+ {\mathcal P}{\mathcal I}\subset {\mathcal I}$.
Finally,  ${\mathcal P}$ is called {\it simple} if $\{{\mathcal
P},{\mathcal P} \}\neq 0$, ${\mathcal P}{\mathcal P}\neq 0$ and
its only ideals are $\{0\}$ and ${\mathcal P}$.

\medskip

We are interested in
 the present paper in studying the structure  of Poisson color algebras ${\mathcal P}$ of arbitrary
 degree.
The paper is organized as follows.  In $\S 2$ we develop
techniques of  connections in the restricted support of ${\mathcal
P}$
      so as to
      show that  ${\mathcal P}$ is of the form ${{\mathcal P}}=U
+ \sum\limits_{j}{\mathcal I}_{j}$ with $U$ a linear subspace of
 ${{\mathcal P}}_0+{{\mathcal P}}_{g_0}+{{\mathcal
P}}_{-g_0}$ and any ${\mathcal I}_j$
     a well described   (graded) ideal of ${\mathcal P}$, satisfying
   $\{{\mathcal I}_j,{\mathcal I}_k\} + {\mathcal I}_j{\mathcal I}_k=0$ if $j\neq k$.
   In $\S 3$, and under mild
   conditions, the   simplicity of ${\mathcal P}$ is characterized and it is
   shown that any Poisson color algebra ${\mathcal P}$ of arbitrary
 degree   is the direct sum of the family of its minimal
   (graded) ideals, each one being a simple Poisson color algebra of
   the same
 degree.

\medskip

   Finally we note that, throughout this paper, Poisson color algebras of
 degree $g_0 \in G$ are considered of arbitrary dimension and over an arbitrary
base  field ${\hu K}$.

\section{Connections and gradings}

In the following,  $${\mathcal P}=\bigoplus\limits_{g \in
G}{{\mathcal P}}_{g}$$ denotes a Poisson color algebra of degree
$g_0$.
 We will write by $$\Sigma=\{g \in G : {\mathcal P}_g \neq 0\}\setminus \{0,\pm g_0\}$$
 the  {\it restricted  support} of ${\mathcal P}$  and by
$$-\Sigma=\{-g: g \in \Sigma\} \subset G \setminus \{0,\pm g_0\}.$$

\begin{definition}\rm \label{connection}
Let $g$ and $h$ be two elements in $\Sigma$. We shall say that $g$
is {\em connected} to  $h$ if there exist
$$\hbox{$g_{1},g_2,...,g_{n} \in \pm \Sigma \cup \{0,\pm g_0\}$ and
$k_2,k_3,...,k_{n} \in \{0,\pm g_0\}$}$$ such that:

\begin{enumerate}
\item[1.] $g_{1} =g,$


 \item[2.]   $g_{1} + g_{2} +k_2 \in \pm\Sigma,$\\
 $g_{1} + g_{2} +k_2+ g_{3} +k_3 \in \pm\Sigma,$\\
$g_{1} + g_{2} +k_2+ g_{3} +k_3+ g_{4}+k_4\in
\pm\Sigma,$\\º$\cdots \cdots \cdots $\\
$g_{1}+ g_{2}+k_2 + g_{3}+k_3+ \cdots +g_{n-1}+k_{n-1}\in
\pm\Sigma,$

\item[4.]  $g_{1}+ g_{2}+k_2 + g_{3}+k_3+ \cdots
+g_{n}+k_{n}=\epsilon h$ for some  $\epsilon \in \pm 1$.
\end{enumerate}

We shall also say that $$\{g_{1}\otimes 0,g_2\otimes k_2,
g_3\otimes k_3,...,g_{n}\otimes k_{n}\}$$ is a {\em connection}
from $g$ to $h$.
\end{definition}

\begin{proposition}\label{pro1}\label{equivalence}
The relation $\sim$ in $\Sigma$, defined by $g \sim h$ if and only
if  $g$ is connected to $h$ is an equivalence relation.
\end{proposition}
\begin{proof}
 The set $\{g \otimes 0\}$ is a connection from $g$ to itself and therefore $g \sim
g $.

 If $g \sim h$ and $\{g_{1}\otimes 0,g_2\otimes k_2,
g_3\otimes k_3,...,g_{n}\otimes k_{n}\}$ is a connection from $g$
to $h$, then it is straightforward  to verify that
$$\{h\otimes 0,-\epsilon g_{n}\otimes -\epsilon k_{n},-\epsilon g_{n-1}\otimes -\epsilon k_{n-1},...,
-\epsilon g_3\otimes -\epsilon k_3,-\epsilon g_2\otimes -\epsilon
k_2\}$$
 is a connection from $h$
to $g$ in case  $$g_{1}+ g_{2}+k_2 + g_{3}+k_3+ \cdots
+g_{n-1}+k_{n-1}+g_{n}+k_{n}= \epsilon h.$$ Therefore $h \sim g$.

 Finally, suppose  $g \sim h$ and
$h \sim l$, and
 write
\begin{equation}\label{conn}
\{g_{1}\otimes 0,g_2\otimes k_2, ...,g_{n}\otimes k_{n}\}
\end{equation}
   for a connection from
$g$ to $h$ and $\{h_1\otimes 0,h_2\otimes k^{\prime}_2,...,
h_{m}\otimes k^{\prime}_{m}\}$ for a connection from $h$ to $l$.

If $m=1$, then $l \in \{\pm h\}$ and so  the own  connection
(\ref{conn})  gives us  $g \sim l$.

If $m> 1$, then it is easy to check that
$$\{g_1\otimes 0,g_2\otimes k_2,..., g_{n}\otimes k_{n},\epsilon h_2 \otimes \epsilon
k^{\prime}_2,\epsilon h_3\otimes  \epsilon
k^{\prime}_3,...,\epsilon h_{m}\otimes \epsilon k^{\prime}_{m}\}$$
is a connection from $g$ to $l$ in case $g_{1}+g_2+k_2+
g_3+k_3+\cdots +g_{n}+k_{n}=\epsilon h.$ Therefore $g$ is
connected to $l$ and $\sim$ is an equivalence relation.
\end{proof}

By Proposition \ref{equivalence} the connection relation is an
equivalence relation in $\Sigma$ and so we can consider the
quotient set
$$\Sigma / \sim=\{[g]: g \in \Sigma\},$$
becoming  $[g]$  the set of elements in the restricted support  of
the grading  which are connected to $g$.

\smallskip

\begin{remark}\label{r1}\rm
Observe that for any $g \in \Sigma$, if $\epsilon g+ \mu g_0 \in
\Sigma$ for some $\epsilon \in \pm1$ and some $\mu \in \{0\} \cup
\{\pm 1\} \cup \{\pm 2\}$ then $$\epsilon g+ \mu g_0 \in [g].$$
Indeed, we just have to consider either the connection $\{g
\otimes 0, 0 \otimes \epsilon \mu g_0 \}$ when $\mu \in \{0\} \cup
\{\pm 1\}$, or $\{g \otimes 0, g_0 \otimes  g_0 \}$ when
$\mu=2\epsilon$, or $\{g \otimes 0, -g_0 \otimes  -g_0 \}$ when
$\mu=-2\epsilon$.
\end{remark}

\bigskip

Our final goal in this section is to associate an adequate
subalgebra ${\mathcal I}_{[g]}$ to any $[g] \in \Sigma / \sim$.

\medskip

Fix $g \in \Sigma$, we start by defining the following linear
subspaces. For any
$$\alpha \in \{0, g_0, -g_0\}$$ let us write
\[
{\mathcal P}_{\alpha,[g]}:=\] \[ \sum\limits_{\{h\in [g], p \in
\Sigma: \hspace{1mm}p=-h-g_0+\alpha\} }\{{\mathcal
P}_{h},{\mathcal P}_{p}\} + \sum\limits_{\{k \in [g], q \in
\Sigma\cup \{-g_0\}: \hspace{1mm} q=-k+\alpha  \}} {\mathcal
P}_{k}{\mathcal P}_{q}\subset {\mathcal P}_{\alpha}.
\]

Observe that whence  $h \in [g]$ and $-h-g_0+ \alpha \in \Sigma$,
(resp. $k \in [g]$ and $-k+\alpha \in \Sigma$),  then the
connection $\{h \otimes 0, g_0 \otimes -\alpha \}$, (resp,  $\{k
\otimes 0, -\alpha \otimes 0 \}$), together with the transitivity
of the connection relation, give us $-h-g_0+ \alpha \in [g]$,
 (resp. $-k+\alpha \in [g]$). Also observe that the possibility $q=-g_0$ just holds when $\alpha=g_0$
 and $k=2g_0 \in \Sigma$.
 Next we define
\[
{\mathcal V}_{[g]}:=\bigoplus_{h \in [g]}{\mathcal P}_{h}.
\]
Finally, we denote by ${\mathcal I}_{[g]}$ the direct sum
\[
{\mathcal I}_{[g]}:= (\sum\limits_{\alpha \in \{0,
g_0,-g_0\}}{\mathcal P}_{\alpha,[g]})\oplus {{\mathcal
V}}_{[g]}.\]

\medskip

\begin{lemma}\label{lemafri-1} For any $g \in \Sigma$ and $h,k \in [g]$ we have
 $\{ {\mathcal P}_{h},{\mathcal P}_{k}\}+ {\mathcal P}_{h}{\mathcal P}_{k}
\subset {\mathcal I}_{[g]}.$
\end{lemma}
\begin{proof}
If $\{ {\mathcal P}_{h},{\mathcal P}_{k}\} \neq 0$ we have two
possibilities. In the first one $h+k+g_0 \in \{0, g_0,-g_0\}$  and
then $\{ {\mathcal P}_{h},{\mathcal P}_{k}\} \subset {\mathcal
P}_{h+k+g_0,[g]}$, or $h+k+g_0 \notin \{0, g_0,-g_0\}$ being then
$\{h\otimes 0, k\otimes g_0\}$ a connection from $h$ to $h+k+g_0
$, that is, $h+k+g_0  \in [g]$ and so  $\{ {\mathcal
P}_{h},{\mathcal P}_{k}\} \subset {\mathcal I}_{[g]}.$

If ${\mathcal P}_{h}{\mathcal P}_{k} \neq 0$, we also have two
cases to distinguish. In the first one $h+k \in  \{0, g_0,-g_0\}$
and so ${\mathcal P}_{h}{\mathcal P}_{k} \subset {\mathcal
P}_{h+k,[g]}$, while in the second one $h+k \notin  \{0,
g_0,-g_0\}$ and then the connection $\{h\otimes 0, k\otimes 0\}$
gives us $h$ is connected to  $h+k$ being   $h+k  \in [g]$.
Consequently   ${\mathcal P}_{h}{\mathcal P}_{k}  \subset
{\mathcal I}_{[g]}.$
\end{proof}

\begin{definition}\rm
For any $\alpha \in \{\pm ng_0: n \in 0,1,2,3\}$ it is said that
${\mathcal P}_{\alpha}$ is {\it tight} if
$${\mathcal P}_{\alpha}=$$ $$\sum\limits_{\{h,p \in \Sigma \setminus \{\pm
ng_0: n \in 2,3\}:p=-h-g_0+\alpha\}}\{{\mathcal P}_h, {\mathcal
P}_{p}\} + \sum\limits_{\{k,q\in \Sigma \setminus \{\pm ng_0: n
\in 2,3\}:q=-k+\alpha\}}{\mathcal P}_k {\mathcal P}_{q}.$$
\end{definition}

\begin{lemma}\label{tight}
If ${\mathcal P}_{g_0}$ is tight then the following assertions
hold.
\begin{enumerate}
\item[{\rm 1.}] If $-2g_0 \in \Sigma$ then $\{{\mathcal
P}_{-2g_0},{\mathcal P}_{g_0} \}  \subset {\mathcal
P}_{0,[-2g_0]}$ and ${\mathcal P}_{-2g_0}{\mathcal P}_{g_0}
\subset {\mathcal P}_{-g_0,[-2g_0]}$.

 \item[{\rm 2.}]  If $-3g_0 \in \Sigma$ then $\{{\mathcal P}_{-3g_0},{\mathcal P}_{g_0} \} \subset {\mathcal
P}_{-g_0,[-3g_0]}.$

\item[{\rm 3.}]  If $-2g_0,g \in \Sigma$ with  $\{{\mathcal
P}_{-2g_0},{\mathcal P}_{0,[g]} \} \neq 0$ and  $G$ is free of
2-torsion, then $[g]=[-2g_0]$ and
$$\{{\mathcal P}_{-2g_0},{\mathcal P}_{0,[g]} \} \subset {\mathcal
P}_{-g_0,[-2g_0]}.$$
\end{enumerate}
\end{lemma}

\begin{proof}
1. Let us begin by showing $\{{\mathcal P}_{-2g_0},{\mathcal
P}_{g_0} \}  \subset {\mathcal P}_{0,[-2g_0]}$. By Jacobi
identity, Leibniz identity and anticommutativity
\begin{equation}\label{gui1}
\{{\mathcal P}_{-2g_0},{\mathcal P}_{g_0} \} \subset$$
$$\sum\limits_{h\in \Sigma \setminus \{\pm ng_0: n \in 2,3\}}
\{{\mathcal P}_{-2g_0},\{{\mathcal P}_h, {\mathcal P}_{-h}\}  \}+
\sum\limits_{k,-k+g_0\in \Sigma \setminus \{\pm ng_0: n \in 2,3\}}
\{{\mathcal P}_{-2g_0}, {\mathcal P}_k {\mathcal P}_{-k+g_0}\}
$$ $$\subset \sum\limits_{h\in \Sigma \setminus \{\pm ng_0: n \in 2,3\}}(\{{\mathcal P}_{h-g_0},{\mathcal P}_{-h } \} +
\{{\mathcal P}_{-h - g_0},{\mathcal P}_{h} \})+$$ $$
\sum\limits_{k,-k+g_0\in \Sigma \setminus \{\pm ng_0: n \in
2,3\}}({\mathcal P}_{k-g_0}{\mathcal P}_{-k+ g_0 } + {\mathcal
P}_{k}\{{\mathcal P}_{-2g_0},{\mathcal P}_{-k+g_0}\}).
\end{equation}
Since for any $p\in \Sigma$ such that  $\epsilon p + \nu g_0
 \notin \{0, g_0,-g_0\}$, where  $\epsilon, \nu \in \{\pm 1\}$, we
 have
 $\epsilon p + \nu g_0 \in \Sigma$ in case ${\mathcal P}_{\epsilon p + \nu g_0} \neq 0$,  and the connection
$\{-2g_0 \otimes 0, \epsilon p \otimes g_0\}$ gives us that in
case ${\mathcal P}_{\epsilon p - g_0} \neq 0$ then $\epsilon p -
g_0 \in [-2g_0]$ for any $\epsilon \in \{ \pm 1\}$, we get that
any
\begin{equation}\label{gui2}
  \{{\mathcal
P}_{h-g_0},{\mathcal P}_{-h } \} + \{{\mathcal P}_{-h -
g_0},{\mathcal P}_{h} \}+ {\mathcal P}_{k-g_0}{\mathcal P}_{-k+
g_0 } \subset {\mathcal P}_{0,[-2g_0]}.
\end{equation}
 Finally, observe that if
$\{{\mathcal P}_{-2g_0},{\mathcal P}_{-k+g_0}\} \neq 0$ then
$\{-2g_0\otimes 0, (-k+g_0) \otimes g_0\}$ is a connection from
$-2g_0$ to $k$ and so
\begin{equation}\label{gui3}
 {\mathcal P}_{k}\{{\mathcal
P}_{-2g_0},{\mathcal P}_{-k+g_0}\} \subset {\mathcal
P}_{k}{\mathcal P}_{-k} \subset {\mathcal P}_{0,[-2g_0]}.
\end{equation}
From Equations (\ref{gui1}), (\ref{gui2}) and (\ref{gui3}) we
complete the assertion.

Let us now prove that ${\mathcal P}_{-2g_0}{\mathcal P}_{g_0}
\subset {\mathcal P}_{-g_0,[-2g_0]}$. By  Leibniz identity and
associativity we get

$${\mathcal P}_{-2g_0}{\mathcal P}_{g_0} \subset$$ $$ \sum\limits_{h\in
\Sigma \setminus \{\pm ng_0: n \in 2,3\}} {\mathcal
P}_{-2g_0}\{{\mathcal P}_h, {\mathcal P}_{-h}\}   +
\sum\limits_{k,-k+g_0\in \Sigma \setminus \{\pm ng_0: n \in 2,3\}}
{\mathcal P}_{-2g_0} ({\mathcal P}_k {\mathcal P}_{-k+g_0})
\subset$$ $$ \sum\limits_{h\in \Sigma \setminus \{\pm ng_0: n \in
2,3\}} (\{{\mathcal P}_{h},{\mathcal P}_{-h-2g_0 } \} + {\mathcal
P}_{h - g_0}{\mathcal P}_{-h} )+ \sum\limits_{k,-k+g_0\in \Sigma
\setminus \{\pm ng_0: n \in 2,3\}} {\mathcal P}_{k-2g_0}{\mathcal
P}_{-k+ g_0 }. $$ Now observe that for any $p \in \Sigma$ such
that $\epsilon p-2g_0 \in \Sigma$, where $\epsilon \in \{\pm 1\}$,
the connection $\{\epsilon p-2g_0\otimes 0, -\epsilon p \otimes
0\}$ gives us $\epsilon p-2g_0 \in [-2g_0]$. From here, in case
$-p+ g_0 \in \Sigma$, we get $\{{\mathcal P}_{p},{\mathcal
P}_{-p-2g_0 } \} + {\mathcal P}_{p-2g_0}{\mathcal P}_{-p+ g_0 }
\subset {\mathcal P}_{-g_0,[-2g_0]}$. Finally, in case ${p - g_0}
\in \Sigma$ we have seen above that $p - g_0 \in [-2g_0]$ and then
${\mathcal P}_{p - g_0}{\mathcal P}_{-p}\subset {\mathcal
P}_{-g_0,[-2g_0]}$ which completes this case.

2. Analogous to the first part of  item 1.

3. We have $$0\neq \{{\mathcal P}_{-2g_0},{\mathcal P}_{0,[g]} \}
\subset  \{{\mathcal P}_{-2g_0},\sum\limits_{\{h\in [g]: -h-g_0
\in \Sigma \}}\{{\mathcal P}_h, {\mathcal P}_{-h-g_0}\}
+\sum\limits_{k\in [g]} {\mathcal P}_k {\mathcal P}_{-k} \}
\subset$$ $$ \sum\limits_{\{h\in [g]: -h-g_0 \in \Sigma
\}}(\{{\mathcal P}_{h-g_0}, {\mathcal P}_{-h-g_0}\}+ \{{\mathcal
P}_{-h-2g_0}, {\mathcal P}_{h}\})+\sum\limits_{k\in [g]}({\mathcal
P}_{k-g_0} {\mathcal P}_{-k}+ {\mathcal P}_{k} {\mathcal
P}_{-k-g_0}).$$  If $\{{\mathcal P}_{h-g_0}, {\mathcal
P}_{-h-g_0}\} \neq 0$ for some $h \in [g]$ with $ -h-g_0 \in
\Sigma $, then the connection $\{h \otimes 0,g_0 \otimes 0, -h
\otimes g_0\}$ gives us $[h]=[-2g_0]$ and so $[g]=[-2g_0]$. Since
$-h-g_0 \in [h]=[-2g_0]$ we also have $ \{{\mathcal P}_{h-g_0},
{\mathcal P}_{-h-g_0}\} \subset {\mathcal P}_{-g_0,[-2g_0]}$ when
$h-g_0 \notin \{0, g_0, -g_0\}$, that is, when $h \neq 2g_0.$ In
case $h=2g_0$ we are dealing with the product $ 0\neq \{{\mathcal
P}_{g_0}, {\mathcal P}_{-3g_0}\}$, but the facts $G$ is free of
2-torsion and $2g_0 \in \Sigma$ show $-3g_0 \in \Sigma$ and so by
item 2 we get $ 0\neq \{{\mathcal P}_{g_0}, {\mathcal P}_{-3g_0}\}
\subset {\mathcal P}_{-g_0,[-3g_0]}= {\mathcal P}_{-g_0,[-2g_0]}$,
where last equality is consequence of  Remark \ref{r1}.

If ${\mathcal P}_{k} {\mathcal P}_{-k-g_0} \neq 0$ for some $k\in
[g]$ we have as in the previous case that if $-k-g_0 \in \Sigma$
then $[g]=[k]=[-2g_0]$,  and consequently ${\mathcal P}_{k}
{\mathcal P}_{-k-g_0} \subset {\mathcal P}_{-g_0,[-2g_0]}$. If
$-k-g_0 \in \{0, g_0,-g_0\}$, then $k=-2g_0$ and so
$[g]=[k]=[-2g_0]$ and we are dealing with the product $0 \neq
{\mathcal P}_{-2g_0} {\mathcal P}_{g_0}.$ But by Item 1, $0 \neq
{\mathcal P}_{-2g_0} {\mathcal P}_{g_0} \subset {\mathcal
P}_{-g_0,[-2g_0]}$.

If  $\{{\mathcal P}_{-h-2g_0}, {\mathcal P}_{h}\} \neq 0$ then in
case $-h-2g_0 \in \Sigma$ we get $[g]=[h]=[-h-2g_0]$ by Remark
\ref{r1}. We observe that  $-h-2g_0 \in \Sigma$. Indeed, in the
opposite case $-h-2g_0 \in \{0, g_0, -g_0\}$ and so $h=-2g_0$, but
then $-h-2g_0=0$ being $-h-2g_0 \in \Sigma,$ a contradiction. From
here, the connection $\{-h-2g_0 \otimes 0, h \otimes 0\}$ shows
$[g]=[h]=[-2g_0]$ and consequently  $0\neq \{{\mathcal
P}_{-h-2g_0}, {\mathcal P}_{h}\} \subset {\mathcal
P}_{-g_0,[-2g_0]}$.

Finally, if ${\mathcal P}_{k-g_0} {\mathcal P}_{-k}\neq 0$, we
have that in case $k-g_0 \in \Sigma$ then $[k-g_0]=[k]=[g]$ by
Remark \ref{r1} and that the connection  $\{k-g_0 \otimes 0, -k
\otimes -g_0\}$ shows $[k-g_0]=[-2g_0]$. Consequently $[g]=[2g_0]$
and $0 \neq {\mathcal P}_{k-g_0} {\mathcal P}_{-k} \subset
{\mathcal P}_{-g_0,[-2g_0]}$. If $k-g_0 \notin \Sigma$ then
$k=2g_0$ and we are dealing with the product $0 \neq {\mathcal
P}_{g_0} {\mathcal P}_{-2g_0}$ which is contained in ${\mathcal
P}_{-g_0,[-2g_0]}$ by Item 1.
\end{proof}

\begin{lemma}\label{dub1} Suppose  $G$ is free of 2-torsion, then for any   $\alpha, \beta  \in \{0, g_0,-g_0\}$ the following assertions hold.
\begin{enumerate}
\item[{\rm 1.}] For   each $h \in \Sigma$ satisfying
$-h-g_0+\alpha \in \Sigma$ we have
\begin{enumerate}
\item[{\rm 1.1.}] if $\alpha + \beta + g_0 \in \{0, g_0,-g_0\}$
then $h+\beta + g_0\in \Sigma$ in case ${\mathcal P}_{h+\beta +
g_0} \neq 0$, and $ -h+\alpha + \beta \in \Sigma$ in case
${\mathcal P}_{-h+\alpha + \beta} \neq 0$.

\item[{\rm 1.2.}] if    $\alpha + \beta  \in \{0, g_0,-g_0\}$ and
${\mathcal P}_{h+\beta+ g_0} \neq 0$ then
 $h+\beta+ g_0 \in
\Sigma \cup \{-g_0\}$.

\item[{\rm 1.3.}] if    $\alpha + \beta  \in \{0, g_0,-g_0\}$ and
$(\alpha, \beta, h) \notin \{ (g_0,-g_0, -2g_0), (0,-g_0, -3g_0)
\}$ then $-h- g_0 + \alpha + \beta  \in \Sigma$ in case ${\mathcal
P}_{-h- g_0 + \alpha + \beta} \neq 0$.

\end{enumerate}

\item[{\rm 2.}] For   each $k \in \Sigma$ satisfying $-k+\alpha
\in \Sigma \cup \{-g_0\}$ we have

\begin{enumerate}

\item[{\rm 2.1.}] if $\alpha + \beta + g_0 \in \{0, g_0,-g_0\}$
and ${\mathcal P}_{k+\beta+g_0} \neq 0$,  then $k+\beta+g_0 \in
\Sigma  \cup \{-g_0\}$.

\item[{\rm 2.2.}] if $\alpha + \beta + g_0 \in \{0, g_0,-g_0\}$
and ${\mathcal P}_{-k+\alpha} \neq 0$,  then $-k+\alpha \in
\Sigma$.

\item[{\rm 2.3.}] if    $\alpha + \beta  \in \{0, g_0,-g_0\}$ and
${\mathcal P}_{-k+  \alpha + \beta} \neq 0$,  then $-k+  \alpha +
\beta  \in \Sigma \cup \{-g_0\}$.
\end{enumerate}
\end{enumerate}
\end{lemma}
\begin{proof}

1.1. Suppose ${\mathcal P}_{h+\beta + g_0} \neq 0$ and $h+\beta +
g_0 \notin \Sigma$ being then $h+\beta + g_0 \in \{0, g_0,
-g_0\}.$

 If $h+\beta + g_0=0$, as $\beta \in \{0,g_0,-g_0\}$ and
$h \in \Sigma$, then necessarily $ \beta=g_0$ and $h= -2g_0$ with
\begin{equation}\label{dub2}
2g_0 \notin \{0, g_0,-g_0\}.
\end{equation}
Since $-h-g_0+\alpha \in \Sigma$ then  $g_0+\alpha \in \Sigma$
and, taking into account $\alpha \in \{0,g_0,-g_0\}$,  we get
$\alpha=g_0$. But we also know $\alpha + \beta + g_0 \in \{0,
g_0,-g_0\}$ and so $3g_0 \in \{0, g_0,-g_0\}$. This implies either
$2g_0 \in \{-g_0, 0 \}$ or $4g_0=0$. In the first case we have a
contradiction with Equation (\ref{dub2}) while in the second one
$0\neq 2g_0$ is an element of $G$ with 2-torsion which is also a
contradiction. A similar argument gives us that the case $h+\beta
+ g_0 \in \pm g_0$ does not hold and so $h+\beta + g_0 \in
\Sigma$. We  can also show as above that $ -h+\alpha + \beta \in
\Sigma$ in case ${\mathcal P}_{-h+\alpha + \beta} \neq 0$.

\smallskip

The remaining items can be proved by arguing as in Item 1.1.
\end{proof}

\begin{lemma}\label{lemafri1}
Suppose  $G$ is free of 2-torsion, then for any $g \in \Sigma$ and
$\alpha, \beta \in \{0, g_0,-g_0\}$ we have
\begin{enumerate}
\item[{\rm 1.}] $\{ {\mathcal P}_{\alpha,[g]},{\mathcal
P}_{\beta,[g]}\} \subset {\mathcal I}_{[g]}.$

\item[{\rm 2.}] If furthermore ${\mathcal P}_{g_0}$ is tight then
${\mathcal P}_{\alpha,[g]}{\mathcal P}_{\beta,[g]}\subset
{\mathcal I}_{[g]}.$
\end{enumerate}
\end{lemma}
\begin{proof}

 1.  
 Suppose there exists $h\in [g]$ with
$-h-g_0+\alpha \in \Sigma$ such that
\begin{equation}\label{fri1}
0 \neq \{\{{\mathcal P}_{h},{\mathcal P}_{-h-g_0+\alpha}\},
{\mathcal P}_{\beta,[g]} \}\subset {\mathcal P}_{\alpha + \beta +
g_0},
\end{equation}
 or there is  $k \in [g]$ with  $-k+\alpha \in \Sigma \cup
\{-g_0\}$
  satisfying
\begin{equation}\label{fri2}
0 \neq \{{\mathcal P}_{k}{\mathcal P}_{-k+\alpha}, {\mathcal
P}_{\beta,[g]} \}\subset {\mathcal P}_{\alpha + \beta + g_0}.
\end{equation}

Let us study Equation (\ref{fri1}), by applying
 Jacobi identity and   anticommutativity  we get
    $$0 \neq \{\{{\mathcal P}_{h},{\mathcal P}_{-h-g_0+\alpha}\}, {\mathcal P}_{\beta} \} \subset$$
    \begin{equation}\label{fri3}
 \{{\mathcal P}_{h+\beta+g_0},{\mathcal P}_{-h-g_0+\alpha}\}+
\{{\mathcal P}_{h},{\mathcal P}_{-h+\alpha+\beta}\} .
\end{equation}
We are going to  distinguish two possibilities, in the first one
$$\alpha+ \beta +g_0 \in \{0, g_0,-g_0\}.$$
From Equation (\ref{fri3}), either $\{{\mathcal
P}_{h+\beta+g_0},{\mathcal P}_{-h-g_0+\alpha}\} \neq 0$ or
$\{{\mathcal P}_{h},{\mathcal P}_{-h+\alpha+\beta}\} \neq 0$. In
the first case,  Lemma \ref{dub1}-1.1 gives us  $h+\beta+g_0
 \in \Sigma.$ Hence,
 taking into account that ${\mathcal P}_{\beta,[g]} \subset
{\mathcal P}_{\beta}$, that the connection $\{h\otimes 0, \beta
\otimes g_0\}$   shows
 $h+\beta+g_0 \in [h]=[g]$, and that $-h-g_0+\alpha \in [g]$  we
 obtain
$ 0 \neq \{{\mathcal P}_{h+\beta+g_0},{\mathcal
P}_{-h-g_0+\alpha}\}\subset {\mathcal P}_{\alpha+\beta+g_0,[g]}.$
In a similar way we have that in case $\{{\mathcal
P}_{h},{\mathcal P}_{-h+\alpha+\beta}\} \neq 0$ then $0 \neq
\{{\mathcal P}_{h},{\mathcal P}_{-h+\alpha+\beta}\} \subset
{\mathcal P}_{\alpha+\beta+g_0,[g]}$ and so we can assert
\begin{equation}\label{porhacer0}
 \{\{{\mathcal P}_{h},{\mathcal P}_{-h-g_0+\alpha}\},
{\mathcal P}_{\beta,[g]} \}\subset {\mathcal
P}_{\alpha+\beta+g_0,[g]}.
\end{equation}
In the second possibility $$\alpha+ \beta +g_0 \notin \{0,
g_0,-g_0\}.$$ We also have from Equation (\ref{fri3}) that either
$\{{\mathcal P}_{h+\beta+g_0},{\mathcal P}_{-h-g_0+\alpha}\}\neq
0$  or $ \{{\mathcal P}_{h},{\mathcal P}_{-h+\alpha+\beta}\} \neq
0$. In the first case, since $-h-g_0+\alpha \in \Sigma$,  the
connection $$\{ h\otimes 0, g_0 \otimes -\alpha, (-h-\beta-g_0)
\otimes -g_0 \}$$ gives us $\alpha+ \beta +g_0 \in [h]=[g]$ while
in the second one the connection $$\{ h\otimes 0, (-h+\alpha
+\beta) \otimes g_0 \}$$ gives us also $\alpha+ \beta +g_0 \in
[h]=[g].$ We have shown $ \{\{{\mathcal P}_{h},{\mathcal
P}_{-h-g_0+\alpha}\}, {\mathcal P}_{\beta,[g]} \}\subset {\mathcal
V}_{[g]}$ in this case and taking also into account Equation
(\ref{porhacer0}) that
$$ \{\{{\mathcal P}_{h},{\mathcal P}_{-h-g_0+\alpha}\},
{\mathcal P}_{\beta,[g]} \}\subset {\mathcal I}_{[g]}.$$

From  Leibniz identity and anticommutativity  we can study
Equation (\ref{fri2}) in a similar way to the above study of
Equation (\ref{fri1}), taking now into account Lemma
\ref{dub1}-2.1. and 2.2. to get
$$ \{{\mathcal P}_{k}{\mathcal P}_{-k+\alpha},
{\mathcal P}_{\beta,[g]} \}\subset {\mathcal I}_{[g]}. $$ and so
we conclude
 $$  \{ {\mathcal P}_{\alpha,[g]},{\mathcal
P}_{\beta,[g]}\}
 \subset {\mathcal I}_{[g]}.$$

2. Suppose  there exists $h\in [g]$ satisfying $-h-g_0+\alpha \in
\Sigma$ and  such that
 $$0 \neq \{{\mathcal P}_{h},{\mathcal
P}_{-h-g_0+\alpha}\} {\mathcal P}_{\beta,[g]} \subset {\mathcal
P}_{\alpha + \beta },$$
 or there is  $k \in [g]$ with  $-k+\alpha \in \Sigma \cup
\{-g_0\}$ such that $$0 \neq ({\mathcal P}_{k}{\mathcal
P}_{-k+\alpha}) {\mathcal P}_{\beta,[g]} \subset {\mathcal
P}_{\alpha + \beta }.$$ An analogous argument to item 1, taking
now into account Lemma \ref{dub1}-1.2., 1.3. and 2.3., and also in
the first possibility that the fact ${\mathcal P}_{g_0}$ tight
together with Remark \ref{r1} imply that in case $-2g_0 \in [g]$
then $\{{\mathcal P}_{-2g_0}, {\mathcal P}_{g_0}\} \in {\mathcal
P}_{0,[g]}$ and that in case $-3g_0 \in [g]$ then $\{{\mathcal
P}_{-3g_0}, {\mathcal P}_{g_0}\} \in {\mathcal P}_{-g_0,[g]}$,
gives us
$$ \{{\mathcal P}_{h},{\mathcal
P}_{-h-g_0+\alpha}\} {\mathcal P}_{\beta,[g]}+ ({\mathcal
P}_{k}{\mathcal P}_{-k+\alpha}) {\mathcal P}_{\beta,[g]} \subset
{\mathcal I}_{[g]} $$ and so $${\mathcal P}_{\alpha,[g]}{\mathcal
P}_{\beta,[g]}\subset {\mathcal I}_{[g]}.$$
\end{proof}

\begin{lemma}\label{lemafri101}
Suppose  ${\mathcal P}_{g_0}$ is tight and $G$ is free of
2-torsion, then for any $g \in \Sigma$, $\alpha \in \{0,
g_0,-g_0\}$ and $k \in [g]$ we have
$$\{ {\mathcal P}_{\alpha,[g]},{\mathcal P}_{k}\} + {\mathcal
P}_{\alpha,[g]}{\mathcal P}_{k} \subset {\mathcal I}_{[g]}.$$
\end{lemma}
\begin{proof}
Suppose $\{ {\mathcal P}_{\alpha,[g]},{\mathcal P}_{k}\} \neq 0$.
We have two cases to distinguish. In the first one $\alpha +k+g_0
\notin \{0, g_0,-g_0\}$ and so $\alpha +k+g_0 \in \Sigma$. Then we
have that the connection $\{k\otimes 0, \alpha \otimes g_0\}$
gives us $\alpha +k+g_0 \in [k]=[g]$ and so $$\{ {\mathcal
P}_{\alpha,[g]},{\mathcal P}_{k}\} \subset {\mathcal V}_{[g]}.$$
In the second case $\alpha +k+g_0 \in \{0, g_0,-g_0\}.$ Taking
also into account $\alpha \in \{0, g_0,-g_0\}$ and $k \notin \{0,
g_0,-g_0\}$ we have that $$(\alpha,k) \in
\{(g_0,-2g_0),(0,-2g_0),(g_0,-3g_0) \}.$$ Consider the possibility
$(\alpha,k) =(g_0,-2g_0)$, that is, $$0\neq \{ {\mathcal
P}_{g_0,[g]},{\mathcal P}_{-2g_0}\} \subset {\mathcal P}_{0}$$
being  $-2g_0 =k\in \Sigma$. From here,  $[-2g_0]=[g]$ and by
Lemma \ref{tight}-1  we get $$0\neq \{ {\mathcal
P}_{g_0,[g]},{\mathcal P}_{-2g_0}\} \subset \{ {\mathcal
P}_{g_0},{\mathcal P}_{-2g_0}\} \subset {\mathcal
P}_{0,[-2g_0]}={\mathcal P}_{0,[g_0]}.$$

In a similar way we can show $0\neq \{ {\mathcal
P}_{\alpha,[g]},{\mathcal P}_{k}\} \subset {\mathcal
P}_{\alpha+g_0+k,[g]}$ when $(\alpha,k) \in
\{(0,-2g_0),(g_0,-3g_0) \}$ and we conclude
$$\{ {\mathcal P}_{\alpha,[g]},{\mathcal P}_{k}\} \subset
{\mathcal I}_{[g]}.$$

\medskip

Suppose now $ {\mathcal P}_{\alpha,[g]}{\mathcal P}_{k} \neq 0$.
If $\alpha+ k \notin \{0, g_0,-g_0\}$ then the connection
$\{k\otimes 0, \alpha \otimes 0\}$ gives us $$0\neq  {\mathcal
P}_{\alpha,[g]}{\mathcal P}_{k} \subset {\mathcal P}_{k+\alpha}
\subset {\mathcal V}_{[g]}.$$  If $\alpha+ k \in \{0, g_0,-g_0\}$
then $(\alpha,k) \in \{(g_0, -2g_0), (-g_0, 2g_0)\}$. If
$(\alpha,k) =(g_0, -2g_0)$ then $0\neq  {\mathcal
P}_{g_0,[g]}{\mathcal P}_{-2g_0} \subset  {\mathcal
P}_{g_0}{\mathcal P}_{-2g_0} \subset {\mathcal P}_{-g_0,[g]}$,
  last inclusion  being consequence of Lemma \ref{tight}-1  and $[-2g_0]=[g]$. Finally, if $(\alpha,k) =(-g_0, 2g_0)$ then $0\neq  {\mathcal
P}_{-g_0,[g]}{\mathcal P}_{2g_0} \subset  {\mathcal
P}_{-g_0}{\mathcal P}_{2g_0} =  {\mathcal P}_{2g_0}{\mathcal
P}_{-g_0} \subset {\mathcal P}_{-g_0,[2g_0]}={\mathcal
P}_{-g_0,[g]}.$ We have shown $$ {\mathcal
P}_{\alpha,[g]}{\mathcal P}_{k} \subset {\mathcal I}_{[g]}.$$
\end{proof}

\begin{proposition}\label{lemasubalge}
Suppose  ${\mathcal P}_{g_0}$ is tight and $G$ is free of
2-torsion, then   for any $g \in \Sigma$ the graded linear
subspace ${\mathcal I}_{[g]}$ is a subalgebra  of ${\mathcal P}$.
\end{proposition}
\begin{proof}
  Since $
{\mathcal I}_{[g]}= (\sum\limits_{\alpha \in \{0,
g_0,-g_0\}}{\mathcal P}_{\alpha,[g]})\oplus {{\mathcal V}}_{[g]}$
we can write
 $$  \{ {\mathcal I}_{[g]}, {\mathcal I}_{[g]}\}
  \subset$$ $$  \sum\limits_{\alpha, \beta \in \{0,
g_0,-g_0\}}\{ {\mathcal P}_{\alpha,[g]},{\mathcal
P}_{\beta,[g]}\}+ \sum\limits_{\alpha \in \{0, g_0,-g_0\}}\{
{\mathcal P}_{\alpha,[g]},
  {{\mathcal V}}_{[g]}\}
 +\{{{\mathcal V}}_{[g]},{{\mathcal V}}_{[g]}\}.$$
 From here,  Lemmas \ref{lemafri-1}, \ref{lemafri1} and
 \ref{lemafri101} allow us to get $  \{ {\mathcal I}_{[g]}, {\mathcal I}_{[g]}\}
  \subset {\mathcal I}_{[g]}.$

 In a similar way we have $   {\mathcal I}_{[g]} {\mathcal I}_{[g]}
  \subset {\mathcal I}_{[g]}$ and consequently  ${\mathcal I}_{[g]}$ is a
  subalgebra of ${\mathcal P}$.
 \end{proof}

\medskip

We call ${\mathcal I}_{[g ]}$ the {\it subalgebra  of ${\mathcal
P}$ associated} to $[g]$.

\section{Decompositions as sum of ideals}

We begin this section by showing that for any $g \in \Sigma$, the
subalgebra ${\mathcal I}_{[g]}$ is actually an ideal of ${\mathcal
P}$. From now on  the group $G$ { will be suppose   free of
2-torsion.}
\begin{proposition}\label{lema1}
If $[g] \neq [h]$ for some $g, h \in \Sigma$ then $\{{\mathcal
I}_{[g]} , {\mathcal I}_{[h]}\}+{\mathcal I}_{[g]}  {\mathcal
I}_{[h]}=0$.
\end{proposition}
\begin{proof}
We have to study   the products
 \[ \{{\mathcal
I}_{[g]}, {\mathcal I}_{[h]}\} =\]  \[\{
 (\sum\limits_{\alpha \in \{0,
g_0,-g_0\}}{\mathcal P}_{\alpha,[g]})\oplus {\mathcal{V}}_{[g]} ,
(\sum\limits_{\alpha \in \{0, g_0,-g_0\}}{\mathcal
P}_{\alpha,[h]})\oplus
 {\mathcal{V}}_{[h]} \}
\]

and

 \[ {\mathcal
I}_{[g]} {\mathcal I}_{[h]} =\]  \[
 ((\sum\limits_{\alpha \in \{0,
g_0,-g_0\}}{\mathcal P}_{\alpha,[g]})\oplus {\mathcal{V}}_{[g]})
((\sum\limits_{\alpha \in \{0, g_0,-g_0\}}{\mathcal
P}_{\alpha,[h]})\oplus
 {\mathcal{V}}_{[h]}).
\]
We begin by considering the summand  $\{{\mathcal{V}}_{[g]} ,
{\mathcal{V}}_{[h]}\}$ of the first product. Suppose  there exist
$k \in [g]$ and $l \in [h]$ such that $0 \neq \{{\mathcal P}_{k} ,
{\mathcal P}_{l}\} \subset {\mathcal P}_{k+l+g_0} .$  We have to
distinguish two cases. In the first one $k+l+g_0 \notin \{0,
g_0,-g_0\} $ and so $k+l+g_0 \in \Sigma$. Then the connection
$\{k\otimes 0, l \otimes g_0,-k\otimes -g_0\}$ gives us
$[g]=[k]=[l]=[h]$, a contradiction. Hence $\{{\mathcal{V}}_{[g]} ,
{\mathcal{V}}_{[h]}\} = 0$ in this case. In the second
possibility, $k+l+g_0 \in \{0, g_0,-g_0\} $. From here $l \in
\{-k, -g_0-k, -2g_0-k \}$ and by Remark \ref{r1} we get $[k]=[l]$,
a contradiction, then
\begin{equation}\label{zoor}
\{{\mathcal{V}}_{[g]} , {\mathcal{V}}_{[h]}\} = 0
\end{equation}
 in any case.

 Consider now the summand  ${\mathcal{V}}_{[g]}
{\mathcal{V}}_{[h]} $ in the second product and suppose  there
exist $k \in [g]$ and $l \in [h]$ such that $0 \neq {\mathcal
P}_{k}  {\mathcal P}_{l} \subset {\mathcal P}_{k+l} .$  In case
$k+l\notin \{0, g_0,-g_0\}$, the connection $\{k\otimes 0, l
\otimes 0, -k \otimes 0 \}$ gives us $[k]=[l]$, a contradiction,
while in case $k+l\in \{0, g_0,-g_0\}$ necessarily $l \in \{-k,
g_0-k, -g_0-k\}$ being then, see Remark \ref{r1}, $[k]=[l]$,  a
contradiction. We have shown
\begin{equation}\label{zoor1}
{\mathcal{V}}_{[g]}  {\mathcal{V}}_{[h]} = 0.
\end{equation}
In order to study the product  $\{\sum\limits_{\alpha \in \{0,
g_0,-g_0\}}
 {\mathcal P}_{\alpha,[g]} ,
 {\mathcal{V}}_{[h]} \},$
 consider any   $$\{\{{\mathcal P}_{k},{\mathcal
P}_{-k-g_0+\alpha}\} , {\mathcal P}_{l}\} $$ with $k \in [g]$
satisfying  $-k-g_0+\alpha \in \Sigma$  and $l \in [h]$. We have
by Jacobi identity and anticommutativity that $$\{\{{\mathcal
P}_{k},{\mathcal P}_{-k-g_0+\alpha}\} , {\mathcal P}_{l}\} \subset
\{\{{\mathcal P}_{k},{\mathcal P}_{l}\} , {\mathcal
P}_{-k-g_0+\alpha}\} + \{\{{\mathcal P}_{-k-g_0+\alpha},{\mathcal
P}_{l}\} , {\mathcal P}_{k}\}.$$ Since by Equation (\ref{zoor})
and Remark \ref{r1} we get $\{{\mathcal P}_{k},{\mathcal
P}_{l}\}=\{{\mathcal P}_{-k-g_0+\alpha} , {\mathcal P}_{l}\}=0$ we
obtain $\{\{{\mathcal P}_{k},{\mathcal P}_{-k-g_0+\alpha}\} ,
{\mathcal P}_{l}\} =0.$ If we now take any  $ \{{\mathcal
P}_{k}{\mathcal P}_{-k+\alpha}, {\mathcal P}_{l} \} $ with $k \in
[g]$  such that  $-k+\alpha \in \Sigma \cup \{-g_0\}$   and $l \in
[h]$ then we get by Leibniz identity and commutativity that
$\{{\mathcal P}_{k}{\mathcal P}_{-k+\alpha}, {\mathcal P}_{l} \}
\subset \{{\mathcal P}_{l}{\mathcal P}_{k}\} {\mathcal
P}_{-k+\alpha}  + {\mathcal P}_{k}\{{\mathcal P}_{l}, {\mathcal
P}_{-k+\alpha} \}$, but by Equation (\ref{zoor}) we have
$\{{\mathcal P}_{l}{\mathcal P}_{k}\}=\{{\mathcal P}_{l},
{\mathcal P}_{-k+\alpha} \}=0$ in case $-k+\alpha \in \Sigma$. If
$-k+\alpha =-g_0$ then
${\mathcal P}_{k}\{{\mathcal P}_{l}, {\mathcal P}_{-k+\alpha} \}
\subset {\mathcal P}_{k} {\mathcal P}_{l} =0$  by Equation
(\ref{zoor1})  being so $ \{{\mathcal P}_{k}{\mathcal
P}_{-k+\alpha}, {\mathcal P}_{l} \}=0 $ in any case. We have
proved
\begin{equation}\label{aa}
\{
 \sum\limits_{\alpha \in \{0,
g_0,-g_0\}}{\mathcal P}_{\alpha,[g]} ,
 {\mathcal{V}}_{[h]} \} + \{
 {\mathcal{V}}_{[g]},\sum\limits_{\alpha \in \{0,
g_0,-g_0\}}{\mathcal P}_{\alpha,[h]} \} =0.
\end{equation}

In a similar way as above, taking now into account Leibniz
identity,  commutativity and associativity we get
\begin{equation}\label{oiuaa}
 (\sum\limits_{\alpha \in \{0,
g_0,-g_0\}}{\mathcal P}_{\alpha,[g]})
 {\mathcal{V}}_{[h]}  +
 {\mathcal{V}}_{[g]}(\sum\limits_{\alpha \in \{0,
g_0,-g_0\}}{\mathcal P}_{\alpha,[h]}) =0.
\end{equation}

 Finally, let us consider  the
case  $\sum\limits_{\alpha, \beta \in \{0, g_0,-g_0\}}\{{\mathcal
P}_{\alpha,[g]}, {\mathcal P}_{\beta,[h]}\}.$ By arguing as in the
previous case, taking now  into account Equation (\ref{aa}) and
the fact (easy to prove) $\{{\mathcal P}_{\beta,[h]}, {\mathcal
P}_{-g_0}\} \subset {\mathcal P}_{\beta,[h]}$, we get
\begin{equation}\label{aaa}
\{
 \sum\limits_{\alpha \in \{0,
g_0,-g_0\}}{\mathcal P}_{\alpha,[g]} , \sum\limits_{\alpha \in
\{0, g_0,-g_0\}}{\mathcal P}_{\alpha,[h]} \} =0.
\end{equation}
In a similar way, by considering now Equations  (\ref{aa}) and
(\ref{oiuaa}), we get
\begin{equation}\label{aaaaa}
 (\sum\limits_{\alpha \in \{0,
g_0,-g_0\}}{\mathcal P}_{\alpha,[g]} )( \sum\limits_{\alpha \in
\{0, g_0,-g_0\}}{\mathcal P}_{\alpha,[h]} ) =0.
\end{equation}

From Equations (\ref{zoor}), (\ref{aa}) and (\ref{aaa}) we get
that $ \{{\mathcal I}_{[g]}, {\mathcal I}_{[h]}\} =0$ while from
Equations (\ref{zoor1}), (\ref{oiuaa}) and (\ref{aaaaa}) that
${\mathcal I}_{[g]} {\mathcal I}_{[h]} =0$ which complete the
proof.
\end{proof}

\begin{theorem}\label{teo1}
Suppose  any ${\mathcal P}_{\alpha}$, $\alpha \in \{0,g_0,-g_0\}$,
is tight then the following assertions hold.
\begin{enumerate}
\item[{\rm 1.}] For any $g \in \Sigma$, the  subalgebra ${\mathcal
I}_{[{g}]} $ of ${\mathcal P}$ associated to  $[{g}]$ is an ideal
of ${\mathcal P}$.
\smallskip
\item[{\rm 2.}] If ${\mathcal P}$ is simple, then there exists a
connection between any two elements of  ${\Sigma}$.
\end{enumerate}
\end{theorem}

\begin{proof}
1. Since   we  can write $\bigoplus\limits_{h \in \Sigma}{\mathcal
P}_h = \bigoplus\limits_{[h] \in \Sigma /\sim}{\mathcal V}_{[h]} $
and ${\mathcal P}_{\alpha} = \sum\limits_{[h] \in \Sigma
/\sim}{\mathcal P}_{\alpha,[h]} $ for any $\alpha \in
\{0,g_0,-g_0\}$,

we have
\begin{equation}\label{cholo}
 {\mathcal P} =  {\mathcal P}_{0}+{\mathcal P}_{g_0}+{\mathcal P}_{-g_0}\oplus (\bigoplus_{h \in
\Sigma}{{\mathcal P}_h})= \sum_{[h] \in \Sigma /\sim}{\mathcal
I}_{[h]}.
\end{equation}
From here,
 by
Propositions \ref{lemasubalge} and \ref{lema1}  we have
$$\{{\mathcal I}_{[{g}]},
{\mathcal P}\}+ \{{\mathcal P},{\mathcal I}_{[{g}]}\} + {\mathcal
I}_{[{g}]} {\mathcal P} + {\mathcal P}{\mathcal I}_{[{g}]} \subset
\{{\mathcal I}_{[{g}]},{\mathcal I}_{[{g}]}\}+\sum\limits_{[h]
\neq [g]} \{{\mathcal I}_{[{g}]}, {\mathcal I}_{[{h}]}\} +
{\mathcal I}_{[{g}]}{\mathcal I}_{[{g}]}+ \sum\limits_{[h] \neq
[g]} {\mathcal I}_{[{g}]} {\mathcal I}_{[{h}]} \subset {\mathcal
I}_{[{g}]}$$ as desired.

\smallskip

(ii) The simplicity of ${\mathcal P}$ applies to get that
${\mathcal I}_{[{g}]}={\mathcal P}$ for any $g \in \Sigma$. Hence
${[{g}]} = \Sigma$ and so
any couple of elements in $\Sigma$  are connected.
\end{proof}

As consequence of Equation (\ref{cholo}), Theorem \ref{teo1} and
Proposition \ref{lema1} we can state the following result.

\begin{theorem} \label{teo2}
Suppose any ${\mathcal P}_{\alpha}$,  $\alpha \in \{0,g_0,-g_0\}$,
is tight. It follows
 $$ {\mathcal P}=
\sum\limits_{[g] \in \Sigma/\sim} {\mathcal I}_{[g]},$$ being any
${\mathcal I}_{[g]}$ one of the ideals given in Theorem
\ref{teo1}.  Moreover, $\{{\mathcal I}_{[g]} , {\mathcal
I}_{[h]}\}+{\mathcal I}_{[g]} {\mathcal I}_{[h]}=0$ whenever $[g]
\neq [h]$.
\end{theorem}

\smallskip

As usual, the {\it center} of ${{{\mathcal P}}}$ is defined as the
set $\{v\in {{{\mathcal P}}}:\{v, {{{\mathcal P}}}\}+ \{
{{{\mathcal P}}},v\} +  v {{{\mathcal P}}} +  {{{\mathcal
P}}}v=0\}$.

\begin{corollary}\label{co1}
If ${\mathcal P}$ is centerless and any ${\mathcal P}_{\alpha}$,
$\alpha \in \{0,g_0,-g_0\}$, is tight then ${\mathcal P}$ is the
direct sums of the ideals given in Theorem \ref{teo1},
\[
{\mathcal P} =\bigoplus_{[g] \in \Sigma/\sim} {\mathcal I}_{[g]}.
\]
\end{corollary}

\begin{proof}

We have to show the direct character of the sum. Given $x\in
{\mathcal I}_{[g]} \cap \sum\limits_{\tiny{\begin{array}{c}
  [h] \in \Sigma / \sim  \\
h \nsim g\\
\end{array}}} {\mathcal I}_{[h]},$
by using  the fact  $\{{\mathcal I}_{[g]} , {\mathcal
I}_{[h]}\}+{\mathcal I}_{[g]} {\mathcal I}_{[h]} = 0$ for $[g]
\neq [h]$ we obtain  $\{x , {\mathcal P}\} + \{ {\mathcal P},x\} +
x{\mathcal P} + {\mathcal P}x= 0$. That is,
$x$ belongs to the center of ${\mathcal P}$ and so $x=0$  as
desired.
\end{proof}


\section{The simple components}


In this section we study if any of the components in the
decomposition given in Corollary \ref{co1} is simple. Under mild
conditions we give an affirmative answer and furthermore a second
Wedderburn-type theorem is stated. Finally,  we recall that in
this section the group $G$ is supposed to be
 free of 2-torsion.

\begin{lemma}\label{lema4}
Let ${\mathcal P}$ be  centerless  and   with ${\mathcal
P}_{\beta}$ tight for $\beta \in \{0, \pm g_0,\pm 2g_0, -3g_0\}$.
 If $I$
is an ideal of ${\mathcal P}$ such that  $I \subset {\mathcal
P}_{0}+ {\mathcal P}_{g_0} + {\mathcal P}_{-g_0}$ then $I=\{0\}$.
\end{lemma}
\begin{proof}
Suppose there exists a nonzero ideal ${\mathcal  I}$ of ${\mathcal
P}$ contained in  $  {\mathcal P}_{0}+ {\mathcal P}_{g_0} +
{\mathcal P}_{-g_0}$.  Since can write
$$I= (I \cap {\mathcal
P}_{0})+ (I \cap {\mathcal P}_{g_0}) + (I \cap {\mathcal
P}_{-g_0}),$$ some $I \cap {\mathcal P}_{\alpha}\neq 0$ for
$\alpha \in \{0, g_0, -g_0\}$. Taking into account ${\mathcal P}$
is centerless, there exists $h \in \Sigma \cup \{0, g_0,-g_0\}$
such that either $\{I \cap {\mathcal P}_{\alpha}, {\mathcal
P}_{h}\} \neq 0$ or $(I \cap {\mathcal P}_{\alpha}) {\mathcal
P}_{h} \neq 0$. In the first case, $0 \neq \{I \cap {\mathcal
P}_{\alpha}, {\mathcal P}_{h}\} \subset {\mathcal
P}_{\alpha+h+g_0} \cap ({\mathcal P}_{0}+ {\mathcal P}_{g_0} +
{\mathcal P}_{-g_0})$ and so necessarily
\begin{equation}\label{gym}
h \in \{0,  g_0, -g_0, -2g_0, -3g_0\},
\end{equation}
but by tightness of the homogeneous spaces associated to these
elements we have  $$0 \neq \{I \cap {\mathcal P}_{\alpha},
{\mathcal P}_{h}\} \subset    \{I \cap {\mathcal P}_{\alpha},
\sum\limits_{{{ p,-p-g_0+h \in \Sigma \setminus \{\pm ng_0: n \in
2,3\}}}}\{{\mathcal P}_p, {\mathcal P}_{-p-g_0+h}\}  +$$
 $$\{I \cap {\mathcal
P}_{\alpha}, \sum\limits_{k,-k+h \in \Sigma \setminus \{\pm ng_0:
n \in 2,3\}} {\mathcal P}_k {\mathcal P}_{-k+h}\} .$$ From here
Jacobi identity and Leibniz identity give us that there exists
some $r \notin \{0, g_0, -g_0, -2g_0, -3g_0\}$ satisfying $ \{I
\cap {\mathcal P}_{\alpha}, {\mathcal P}_{r}\} \neq 0$ which
contradicts Equation (\ref{gym}).

In the second case $0 \neq (I \cap {\mathcal P}_{\alpha})
{\mathcal P}_{h} \subset {\mathcal P}_{\alpha+h} \cap ({\mathcal
P}_{0}+ {\mathcal P}_{g_0} + {\mathcal P}_{-g_0})$. This fact only
occurs for
$$h \in \{0,  g_0, -g_0, 2g_0, -2g_0\}.$$
A similar above argument with the tightness of the homogeneous
spaces associated to these elements, Leibniz identity and
associativity gives us a contradiction. Hence we conclude ${
I}=0$.
\end{proof}

 Let us introduce
the concepts of   maximal length and $\Sigma$-multiplicativity in
the setup of  Poisson color  algebras of degree $g_0$ in a similar
way than in the frameworks of graded Lie algebras, graded Lie
superalgebras, graded Leibniz algebras, split Poisson algebras,
split color Lie algebras etc. (see \cite{At1, Poissonyo, coloryo,
At2, Kochetov} for discussion and examples on these concepts).


\begin{definition}\label{Def30}\rm
We say that a  Poisson color algebra ${\mathcal P}$  of degree
$g_0$ is of {\it maximal length} if ${\mathcal P}_0 \neq 0$ and
$\dim {{\mathcal P}}_{g} =1$ for any $g\in \Sigma$.
 \end{definition}

\begin{definition}\rm
We say that a   Poisson color algebra ${\mathcal P}$  of degree
$g_0$  is {\it $\Sigma$-multiplicative} if given $g\in \Sigma $
and $ h \in \Sigma \cup \{0, \pm g_0\}$
 such that $ g + h +k\in \Sigma $ for some $k \in \{0, g_0,-g_0\}$ then
 $ {{\mathcal P}}_{g}  { {\mathcal P}}_{h}\neq 0$ if $k=0$,
 $ \{{{\mathcal P}}_{g} , { {\mathcal P}}_{h}\}\neq 0$ if $k=g_0$
 or $ ({\mathcal P}_{g}  { {\mathcal P}}_{h}) { {\mathcal P}}_{-g_0}\neq 0$ if $k=-g_0.$
  \end{definition}
We recall that $\Sigma$  is called {\it symmetric} if
  $g\in \Sigma$ implies $-g\in \Sigma.$ From now on we will
  suppose $\Sigma$ is symmetric.

  \medskip

   We would like to note that the above concepts appear in a natural
way in the study of any  Poisson system. For instance, any graded
Poisson structure associate to the  Cartan grading of a semisimple
finite dimensional Lie algebra gives rise to a
$\Sigma$-multiplicative   graded Poisson algebra with symmetric
support and of maximal length. We also have, in the
infinite-dimensional setting, that any graded Poisson structure
${\mathcal P}$ defined either on the split grading of a semisimple
separable $L^*$-algebra, \cite{Schue1, Schue2}, or on a semisimple
locally finite split Lie algebra, \cite{Stumme}, necessarily makes
$\mathcal{P}$ a  graded Poisson algebra with symmetric support,
$\Sigma$-multiplicative and  of maximal length. The Poisson
algebras considered in \cite[\S3]{Poissonyo} are also examples of
 graded Poisson algebras with symmetric support  of maximal length
and $\Sigma$-multiplicative.


\begin{lemma}\label{copa} Let ${\mathcal P}$ be  centerless,  $ \Sigma$-multiplicative,   of maximal length and  with ${\mathcal
P}_{\beta}$ tight for $\beta \in \{0, \pm g_0,\pm 2g_0, -3g_0\}$.
If any couple of elements in $\Sigma$ are connected, then any
nonzero ideal $I$ of ${\mathcal P}$ satisfies $I={\mathcal P}$.
\end{lemma}
\begin{proof}
Consider $I$  a nonzero ideal of ${\mathcal P}$ and write  $I=
(I\cap {\mathcal P}_0 )+(I\cap {\mathcal P}_{g_0} )+(I\cap
{\mathcal P}_{-g_0} )\oplus (\bigoplus\limits_{g \in \Sigma_I} (I
\cap {\mathcal P}_{g}))$ where $\Sigma_{I}:=\{g \in \Sigma: I \cap
{\mathcal P}_{g} \neq 0\}$. By  the maximal length of ${\mathcal
P}$
 we can write
$$  I= (I\cap {\mathcal P}_0 )+(I\cap {\mathcal P}_{g_0} ) +(I\cap
{\mathcal P}_{-g_0} ) \oplus (\bigoplus\limits_{g \in \Sigma_I}
{\mathcal P}_{g}),$$ being  $\Sigma_I \neq \emptyset$ as
consequence of Lemma \ref{lema4}. From here,   we can
 take $g \in \Sigma_I$ being so
\begin{equation}\label{I}
 0 \neq {\mathcal P}_{g}
\subset I.
\end{equation}
 For any  $h \in \Sigma$, $h \neq \pm g$,  the fact that
 ${g}$ and $h$ are connected
 allows  us to fix a
connection $$\{g_1\otimes 0,g_2 \otimes k_2,....,g_n \otimes
k_n\}$$ from ${g}$ to $h$. Consider $g_1=g$, $g_2$ and $ g_1 +
g_2+k_2$.
 By $\Sigma$-multiplicativity and  maximal length
of ${\mathcal P}$ we obtain either $0\neq {\mathcal P}_{g_1}
 {\mathcal P}_{g_2} ={\mathcal
P}_{g_1+g_2}$ if $k_2=0$ or $0\neq \{ {\mathcal P}_{g_1} ,
 {\mathcal P}_{g_2} \}={\mathcal
P}_{g_1+g_2+g_0}$ if $k_2=g_0$ or $0\neq ({\mathcal P}_{g_1}
 {\mathcal P}_{g_2}){\mathcal P}_{-g_0} ={\mathcal
P}_{g_1+g_2-g_0}$ if $k_2=-g_0$. From here, Equation (\ref{I})
gives us  that in any case $$0\neq {\mathcal P}_{g_1+g_2+k_2}
\subset I.$$
  We
can argue in a similar way from  $ g_1+g_2+k_2$, $g_3$ and $
g_1+g_2+k_2+g_3+k_3$ to get
$$0\neq {\mathcal P}_{g_1+g_2+k_2+g_3+k_3} \subset I.$$
 Following this process with the
connection $\{g_1\otimes 0,g_2 \otimes k_2,....,g_n \otimes k_n\}$
we obtain that
$$0\neq {\mathcal P}_{g_1+g_2+k_2+g_3+k_3+ \cdots +g_n+k_n} \subset I$$ and so either $ {\mathcal P}_{ h}
\subset I$ or  $ {\mathcal P}_{ -h} \subset I$.  That is, $0\neq
{\mathcal P}_{ \epsilon_h h} \subset I$  for any $h \in \Sigma$
and some $\epsilon_h \in \{\pm 1\}$.

  Now, observe that we have showed that in case $h \notin \Sigma_I$ for some $h
  \in \Sigma$, then $-h \in \Sigma_I$. From here, if $-h+g_0 \in
  \Sigma$, (resp. $-h-g_0 \in
  \Sigma$, $-h-2g_0 \in
  \Sigma$), then by considering the set $-h, g_0, 0$, (resp.  $-h, -g_0,
  0$; $-h, -g_0,
  -g_0$), the $\Sigma$-multiplicativity and  maximal length
of ${\mathcal P}$ give us now $ {\mathcal P}_{ -h+g_0} \subset I$,
(resp. $ {\mathcal P}_{- h-g_0} \subset I$, $ {\mathcal P}_{
-h-2g_0} \subset I$). Hence,  the fact  ${\mathcal P}_{\alpha}$ is
tight for any $\alpha \in \{0,g_0,-g_0\}$
  allows us to assert
${\mathcal P}_{0}+{\mathcal P}_{g_0} + {\mathcal P}_{-g_0}\subset
I. $

Finally, the $\Sigma$-multiplicativity and  maximal length of
${\mathcal P}$ together with the fact $ {\mathcal P}_{0}\subset I
$ allow us to assert that ${\mathcal P}_{h}={\mathcal
P}_{h}{\mathcal P}_{0} \subset I$ for any $h \in \Sigma$. Since
$${\mathcal P}={\mathcal P}_0 + {\mathcal P}_{g_0} + {\mathcal P}_{-g_0} \oplus
(\bigoplus\limits_{h \in \Sigma} {\mathcal P}_{ h}) \subset I$$
the proof is completed.
\end{proof}

As consequence of Theorem \ref{teo1}-2 and Lemma \ref{copa} we can
assert the next result.

\begin{theorem}\label{teo3} Let ${\mathcal P}$ be  centerless,  $ \Sigma$-multiplicative,   of maximal length and  with ${\mathcal
P}_{\beta}$ tight for $\beta \in \{0, \pm g_0,\pm 2g_0, -3g_0\}$.
Then ${\mathcal P}$ is simple if and only if it has any couple of
elements in $\Sigma$  connected.
\end{theorem}

\begin{theorem}
Let ${\mathcal P}$ be  centerless,  $ \Sigma$-multiplicative,   of
maximal length and  with ${\mathcal P}_{\beta}$ tight for $\beta
\in \{0, \pm g_0,\pm 2g_0, -3g_0\}$. Then ${\mathcal P}$ is the
direct sum of the family of its minimal ideals, each one being a
simple  Poisson color algebra of degree $g_0$ having all of the
elements in its restricted support  connected.
\end{theorem}

\begin{proof}
By Corollary \ref{co1} we have that ${\mathcal P}
=\bigoplus\limits_{[g] \in \Sigma/\sim} {\mathcal I}_{[g]}$ is the
direct sum of the ideals $ {\mathcal I}_{[g]}.$
We wish to apply Theorem \ref{teo3} to any ${\mathcal I}_{[g]}$,
so we have to verify that ${\mathcal I}_{[g]}$ is a centerless $
\Sigma$-multiplicative   Poisson color algebra of degree $g_0$ with
maximal length, with  $({{\mathcal I}_{[g]}})_{\beta}$
   tight for $\beta \in  \{0, \pm g_0,\pm 2g_0, -3g_0\}$ and with all of the elements in its restricted support
   connected.

Since   $({{\mathcal I}_{[g]}})_{\beta}={\mathcal
P}_{{\beta},[g]}$ in case  $\beta \in  \{0, \pm g_0\}$ and
${\mathcal P}_{{\beta}}$ is tight,  we clearly have $({{\mathcal
I}_{[g]}})_{\beta}$ is tight for $\beta \in \{0, \pm g_0\}$. In
case  $\beta \in  \{\pm 2g_0, -3g_0\} \setminus \{0, \pm g_0\}$
with $\beta \in \Sigma$, then $\beta \in [k]$ for a unique $[k]
\in \Sigma/\sim$ and so $({{\mathcal I}_{[g]}})_{\beta}=0$ if $[g]
\neq [k]$ and $({{\mathcal I}_{[k]}})_{\beta}= \sum\limits_{h\in
[k] \setminus \{\pm ng_0: n \in 2,3\}}(\{{\mathcal P}_h, {\mathcal
P}_{-h-g_0+\beta}\} + {\mathcal P}_h {\mathcal P}_{-h+\beta})$.
From here, taking into account Remark \ref{r1},  $({{\mathcal
I}_{[g]}})_{\beta}$ is tight in any case.

We also have ${\mathcal I}_{[g]}$ is $ \Sigma$-multiplicative as
consequence of the $ \Sigma$-multiplicativity of ${\mathcal P}$
and clearly ${\mathcal I}_{[g]}$ is of maximal length. Also
observe that ${\mathcal I}_{[g]}$ is centerless  as consequence of
the fact $\{{\mathcal I}_{[g]} , {\mathcal I}_{[h]}\}+{\mathcal
I}_{[g]}  {\mathcal I}_{[h]}=0$ if $[g] \neq [h]$, (Theorem
\ref{teo2}), and that $ {\mathcal P}$ is centerless. Finally,
since the restricted support of ${\mathcal I}_{[g]}$
   is $[g]$,  it  is
easy to verify that  $[g]$  has all of its elements
$[g]$-connected, (connected through elements contained in  $[g]
\cup \{0, \pm g_0\}$).
 From the above, we can apply
Theorem \ref{teo3} to any ${\mathcal I}_{[g]}$ so as to conclude
${\mathcal I}_{[g]}$ is simple. It is clear  that the
decomposition
 ${\mathcal P} =\bigoplus\limits_{[g] \in
\Sigma/\sim} {\mathcal I}_{[g]}$
 satisfies the assertions of the theorem.
\end{proof}


\

{\sc Antonio J. Calder\'{o}n Martin.}  Department of Mathematics.
Univesity of  C\'{a}diz. 11510 Puerto Real, C\'{a}diz. (Spain),
e-mail: \emph{ajesus.calderon@uca.es}

\

{\sc D. Mame Cheikh.} Department of Mathematics. University of
Dakar.  e-mail: \emph{dioufmamecheikh@gmail.com}


\begin{thebibliography}{99}


\bibitem{G1} Akman, F.:  A master identity for homotopy Gerstenhaber algebras.
Comm. Math. Phys. 209, no. 1,  (2000), 51--76.



\bibitem{2}  Avan, J. and  Doikou, A.: Boundary Lax pairs from
non-ultra-local Poisson algebras. J. Math. Phys. 50, no. 11,
(2009), 113512, 9 pp.


\bibitem{Ca4} Batalin, I.A. and Fradkin, E.S.: Quantization of
gauge theories with linearly dependent generators.  Phys. Rev. D
(3), 28(10),  (1983), 2567--2582.

\bibitem{C1} Bautista, C.  A Poincaré-Birkhoff-Witt theorem for generalized
Lie color algebras. J. Math. Phys. 39, no. 7,  (1998), 3828–3843.

\bibitem{S1} Bruce, A.J.: From $L_{\infty}$-algebroids to higher Schouten/Poisson
structures. Rep. Math. Phys. 67 , no. 2, (2011),  157--177.




 \bibitem{At1} Calder\'{o}n, A.J.: On the structure of graded Lie
algebras. J. Math. Phys. 50, no. 10, (2009), 103513, 8 pp.

 \bibitem{Poissonyo} Calder\'{o}n, A.J.: On the structure of  split non-commutative Poisson algebras.
 Linear and multilinear algebra.  60(7),  (2012), 775--785.




\bibitem{coloryo} Calderó\'{o}n, A.J. and S\'{a}nchez, J.M.: On the structure of split  Lie color
algebras.    Linear Algebra and its applications. 436(2),
 (2012), 307--315.


\bibitem{At2} Calder\'{o}n, A.J. and S\'{a}nchez, J.M.: Weight modules over split Lie algebras.
{Modern Phys. Letters A. } 28(5), (2013), 1350008 9 pp.


\bibitem{Catta} Cattaneo, A.S., Fiorenza D. and  Longoni R.: Graded
Poisson algebras. In Francoise, J.P.,  Naber, G.L.; Tsun,T. S.
Encyclopedia of Mathematical Physics. Amsterdam, 560--567.


\bibitem{super} Draper, C., Mart\'{\i}n, C. and Elduque. A.:  Fine gradings on exceptional simple Lie superalgebras. Internat. J. Math. 22(12), (2011),
1823--1855.

\bibitem{S2} Gibbons, J.,  Holm,  D. and  Tronci, D.:  Cesare Geometry of Vlasov
kinetic moments: a bosonic Fock space for the symmetric Schouten
bracket. Phys. Lett. A 372, no. 23, (2008, 4184–-4196.



\bibitem{Ca13} Henneaux M. and Teitelboim, C.: Quantization of gauge
systems. Princenton University Press, Princenton, NJ, 1992.


\bibitem{4}  Hone, A.N. and  Petrera, M.: Three-dimensional discrete
systems of Hirota-Kimura type and deformed Lie-Poisson algebras.
J. Geom. Mech. 1, no. 1, (2009), 55--85.


\bibitem{Geo2} Ikeda, N.: Deformation of graded Poisson (Batalin-Vilkovisky)
structures. Poisson geometry in mathematics and physics, 147--161,
Contemp. Math., 450, Amer. Math. Soc., Providence, RI, 2008.

\bibitem{C2}  Khakimdjanov, Y. and  Navarro, R.M.: Integrable deformations of
nilpotent color Lie superalgebras. J. Geom. Phys. 61 , no. 10,
(2011), 1797--1808.

\bibitem{Kochetov} Kochetov, M.: Gradings on finite-dimensional simple Lie algebras. Acta Appl.
Math. 108, no. 1,  (2009), 101--127.




\bibitem{Geo5} Kosmann-Schwarzbach, Y.: Poisson and symplectic functions
in Lie algebroid theory. Higher structures in geometry and
physics, 243--268, Progr. Math., 287, Birkhäuser/Springer, New
York, 2011.


\bibitem {F2} Ngakeu, F.:  Graded Poisson structures and Schouten-Nijenhuis
bracket on almost commutative algebras. Int. J. Geom. Methods Mod.
Phys. 9, no. 5, (2012), 1250042, 20 pp.



\bibitem{Cor} Paufler, C.: A vertical exterior derivate in
multisymplectic geometry and graded Poisson bracket for nontrivial
geometries. Rep. Math. Phys.47(1), (2001), 101-119


  \bibitem {Dmitri} Piontkovski, D. and  Silvestrov, S.D.: Cohomology of 3-dimensional color Lie algebras.
  J. Algebra 316, no. 2, (2007), 499--513.


  \bibitem {Price}  Price, K.L.: A domain test for Lie color algebras. J. Algebra Appl. 7, no. 1, (2008), 81--90.


 \bibitem {1} Rottenberg, V.D.: Generalized superalgebras. Nucl.
 Phys. B. 139, (1978), 189-202.

 \bibitem{Schue1}   Schue J.R.: Hilbert Space methods in the theory of
Lie algebras. Trans. Amer. Math. Soc.  95, (1960), 69-80.

\bibitem{Schue2}   Schue J.R.:  Cartan decompositions for
$L^{*}$-algebras.
 Trans. Amer. Math. Soc. 98 (1961),  334--349.





\bibitem {Stumme}  Stumme N.:  The structure of Locally Finite split Lie
Algebras.  J.  Algebra.     220 (1999), 664--693.




 \bibitem {JMP2} Su, Y., Zhao, K. and  Zhu, L.: Classification of
derivation-simple color algebras related to locally finite
derivations. J. Math. Phys. 45, no. 1, (2004), 525--536.

\bibitem {F1} Trostel, R.: Color analysis, variational selfadjointness, and
color Poisson (super)algebras. J. Math. Phys. 25, no. 11, (1984),
3183--3189.


\bibitem {G2}  Vankerschaver, J.,  Yoshimura, H. and  Leok, M.:  On the
geometry of multi-Dirac structures and Gerstenhaber algebras. J.
Geom. Phys. 61 , no. 8, (2011), 1415--1425.


\bibitem {JMP1} Wei, Y.  and Wettig, T.: Bosonic color-flavor transformation for the
special unitary group. J. Math. Phys. 46, no. 7, (2005), 072306,
25 pp.

\bibitem {G3} Xu, P.  Gerstenhaber algebras and $BV$-algebras in Poisson
geometry. Comm. Math. Phys. 200, no. 3,  (1999), 545--560.



\bibitem {Zhang} Zhang, Q. and   Zhang, Y.: Derivations and extensions of
Lie color algebra. Acta Math. Sci. Ser. B Engl. Ed. 28, no. 4,
(2008),  933--948.





 \bibitem {Xueme} Zhang, X. and  Zhou, J.: Derivation simple color algebras and semisimple Lie color algebras.
   Comm. Algebra 37, no. 1, (2009), 242--257.







\bibitem {Kaiming} Zhao, K.: Simple Lie color algebras from graded associative
algebras. J. Algebra 269, no. 2, (2003), 439--455






\end{thebibliography}
\end{document}